\documentclass[english, oribibl]{llncs}

\usepackage{llncsdoc}

\usepackage[T1]{fontenc}
\usepackage[latin9]{inputenc}
\usepackage{hyperref}
\usepackage{verbatim}
\usepackage{amsfonts}
\usepackage{amsmath}
\usepackage{multirow}
\usepackage{url}

\usepackage{graphicx}
\usepackage{amsfonts}
\usepackage{courier}
\usepackage[usenames]{color}
\usepackage{bm}
\usepackage{times}
\newenvironment{sketch}{\noindent{\bf Proof [Sketch]:~~}}{\(\qed\)}

\newcommand{\Exp}{\operatornamewithlimits{\mathbb{E}}}

\newcommand{\mcA}{{\mathcal A}}

\newcommand{\ith}[1]{{#1}^{\text{th}}}

\newcommand{\mygamma}{L}

\newcommand{\ignore}[1]{}

\makeatother

\usepackage{babel}
\makeatother

\begin{document}

\title{Converting Online Algorithms to Local Computation Algorithms}
\author{
Yishay Mansour\inst{1} 
\thanks{\email{mansour@post.tau.ac.il}. Supported in part by the Google Inter-university center
for Electronic Markets and Auctions, by a grant from the Israel
Science Foundation, by a grant from United States-Israel Binational
Science Foundation (BSF), by a grant from Israeli Centers of Research Excellence (ICORE), and by a grant from the Israeli Ministry
of Science (MoS).}
\and
Aviad Rubinstein\inst{1}
\thanks{ \email{aviadrub@mail.tau.ac.il}.}
\and
Shai Vardi\inst{1}
\thanks{\email{shaivar1@post.tau.ac.il}. Supported in part by the Google Inter-university center
for Electronic Markets and Auctions }
\and
Ning Xie\inst{2}
\thanks{\email{ningxie@csail.mit.edu}. Supported by NSF grants CCF-0728645, CCF-0729011 and CCF-1065125.}
 \institute{School of Computer Science, Tel Aviv University, Israel
\and CSAIL, MIT, Cambridge MA 02139, USA}
}

\date{}

\setcounter{page}{1}
\setcounter{secnumdepth}{3}
\maketitle

\begin{abstract}

We propose a general method for converting online algorithms to local computation algorithms,\footnote{For a given input $x$, \emph{local computation algorithms} support queries by a user to values of specified locations $y_i$ in a legal output $y \in F(x)$.} by selecting a random permutation of the input, and simulating running the online algorithm. We bound the number of steps of the algorithm using a \emph{query tree}, which models the dependencies between queries. We improve previous analyses of query trees on graphs of bounded degree, and extend the analysis to the cases where the degrees are distributed binomially, and to a special case of bipartite graphs.

Using this method, we give a local computation algorithm for maximal matching in graphs of bounded degree, which runs in time and space $O(\log^3{n})$.

We also show how to convert a large family of load balancing algorithms (related to balls and bins problems) to local computation algorithms. This gives several local load balancing algorithms which achieve the same approximation ratios as the online algorithms, but run in $O(\log{n})$ time and space.

Finally, we modify existing local computation algorithms for hypergraph $2$-coloring and $k$-CNF and use our improved analysis to obtain better time and space bounds, of $O(\log^4{n})$, removing the dependency on the maximal degree of the graph from the exponent.
\end{abstract}

\section{Introduction}
\label{section:introduction}

\subsection{Background}
The classical computation model has a single processor which has
access to a given input, and using an internal memory, computes the
output. This is essentially the von Newmann architecture, which has
been the driving force since the early days of computation.
The class of polynomial time algorithms is widely accepted as the
definition of {\em efficiently computable} problems.
Over the years many interesting variations of this basic model have
been studied, focusing on different issues.


Online algorithms (see, e.g., \cite{BEY98}) introduce limitations in the time domain. An online
algorithm needs to select actions based only on the history it
observed, without access to future inputs that might influence its
performance. Sublinear algorithms (e.g. \cite{MR06, PR07}) limit the space domain, by
limiting the ability of an algorithm to observe the entire input,
and still strive to derive global properties of it.

Local computation algorithms (LCAs) \cite{RTVX11} are a variant of
sublinear algorithms.
The LCA model considers a computation problem which might have
multiple admissible solutions,  each consisting of multiple bits. The
LCA can return queries regarding parts of the output, in a consistent
way, and in poly-logarithmic time. For example, the input for an LCA
for a job scheduling problem consists of the description of $n$ jobs
and $m$ machines. The admissible solutions might be the allocations
of jobs to machines such that the makespan is at most twice the optimal
makespan. On any query of a job, the LCA answers quickly the job's machine.
The correctness property of the LCA guarantees that different query
replies will be consistent with some admissible solution.



\subsection{Our results}

Following  \cite{ARV+11}, we use an abstract tree structure - \emph{query trees} to bound the number of queries performed by certain algorithms. We use these bounds to improve the upper bound the time and space requirements of several algorithms introduced in \cite{ARV+11}. We also give a generic method of transforming online algorithms to LCAs, and apply it to obtain LCAs to maximal matching and several load balancing problems.

\subsubsection{Bounds on query trees}

Suppose that we have an online algorithm where the reply to a query depends
on the replies to a small number of previous queries. The reply to each of those previous queries depends on the replies to a small number of other queries and so on. These dependencies can be used to model certain problems using \emph{query trees} -- trees which model the dependency of the replies to a given query on the replies to other queries.

Bounding the size of a query tree is central to the analyses of our algorithms.  We show that the size of the query tree is  $O(\log{n})$ w.h.p., where $n$ is the number of vertices. $d$, the degree bound of the dependency graph, appears in the constant. \footnote{Note that, however, the hidden constant is exponentially dependent on $d$. Whether or not this bound can be improved to have a polynomial dependency on $d$ is an interesting open question.} This answers in the affirmative the conjecture of \cite{ARV+11}.
Previously, Alon et al. \cite{ARV+11} show that the expected size of the query tree is constant, and  $O(\log^{d+1}{n})$ w.h.p.\footnote{Notice that bounding the expected size of the query tree is not enough for our applications, since in LCAs we need to bound the probability that \emph{any} query fails.} Our improvement is significant in removing the dependence on $d$ from the exponent of the logarithm. 
 We also show that when the degrees of the graph are distributed binomially, we can achieve the same bound on the size of the query tree. In addition, we show a trivial lower bound of $\Omega(\log{n}/\log{\log{n}})$.
 
We use these results on query trees to obtain LCAs for several online problems -- maximal matching in graphs of bounded degree and several load balancing problems. We also use the results to improve the previous algorithms for hypergraph $2$-coloring and $k$-CNF. 

\subsubsection{Hypergraph $2$-coloring}
We modify the algorithm of \cite{ARV+11} for an LCA for hypergraph $2$-coloring, and coupled with our improved analysis of query tree size, obtain an LCA which runs in time and space $O(\log^4{n})$, improving the previous result, an LCA which runs $O(\log^{d+1} n)$  time and space. 

\subsubsection{$k$-CNF}
Building on the similarity between hypergraph $2$-coloring and $k$-CNF, we  apply our results on hypergraph $2$-coloring to give an an LCA for $k$-CNF which runs in time and space $O(\log^4{n})$.\\

We use the query tree to transform online algorithms to LCAs. We simulate online algorithms as follows: first a random permutation
of the items is generated on the fly. Then, for each query, we simulate the online algorithm on a stream of input items arriving according to
the order of the random permutation. Fortunately, because of the nature of our graphs (the fact that the degree is bounded or distributed binomially), we show that in expectation, we will only need to query a constant number of nodes, and only $O(\log{n})$ nodes w.h.p. We now state our results:

\subsubsection{Maximal matching}
We simulate the greedy online algorithm for maximal matching, to derive an LCA for maximal matching which runs in time and space $O(\log^3{n})$.

\subsubsection{Load Balancing}
 We give several LCAs to load balancing problems which run in $O(\log{n})$ time and space. Our techniques include extending 
 the analysis of the query tree size to  the case where the degrees are
selected from a binomial distribution with expectation $d$, and further extending it to bipartite graphs which exhibit 
the characteristics of many balls and bins problems, specifically ones where each ball chooses $d$ bins at random. We show how to convert a large class of the ``power of $d$ choices'' 
 online algorithms (see, e.g., \cite{ABK+99, BCS+06, TW07}) to efficient
LCAs.

\subsection{Related work}
Nguyen and Onak \cite{NO08} focus on transforming classical approximation algorithms into constant-time algorithms that approximate the size of the optimal solution of problems such as vertex cover and  maximum matching. They generate a random number $r \in [0, 1]$, called the rank, for each node. These ranks are used to bound the query tree size. 

Rubinfeld et al. \cite{RTVX11} show how to construct polylogarithmic time local computation algorithms to maximal independent set computations, scheduling radio network broadcasts, hypergraph coloring and satisfying k-SAT formulas. Their proof technique uses Beck's analysis in
his algorithmic approach to the Lov{\'{a}}sz Local Lemma \cite{Bec91}, and a reduction from distributed algorithms. 
Alon et al. \cite{ARV+11}, building on the technique of \cite{NO08}, show how to extend several of the algorithms of \cite{RTVX11} to perform in polylogarithmic space as well as time. They further observe that we do not actually need to assign each query  a rank, we only need a random permutation of the queries. Furthermore, assuming the query tree is bounded by some $k$, the query to any node depends on at most $k$ queries to other nodes, and so a $k$-wise independent random ordering suffices. They show  how to construct a $1/n^2$-almost $k$-wise independent random ordering\footnote{A random ordering $D_r$ is said to be \emph{$\epsilon$-almost $k$-wise independent} if the statistical distance between $D_r$ and some $k$-wise independent random ordering by at most $\epsilon$.} from a seed of length $O(k \log^2{n})$. 



Recent developments in sublinear time algorithms for sparse graph and
combinatorial optimization problems have led to new constant time
algorithms for approximating the size of a minimum vertex cover, maximal matching, 
maximum matching, minimum dominating set, and other problems (cf. \cite{PR07, MR06, NO08, YYI09}), by 
randomly querying a constant number of vertices.
A major difference between these algorithms and LCAs is that LCAs require that w.h.p., the 
output will be correct on any input, while optimization problems usually require a correct output only on \emph{most} inputs. More importantly, LCAs reuire a consistent output for each query, rather than only approximating a given global property.

There is a vast literature on the topic of balls and bins and the power of $d$ choices. (e.g. \cite{ABK+99, BCS+06, DR96, TW07}). For a survey on the power of $d$ choices, we refer the reader to \cite{MRR01}.

\subsection{Organization of our paper}

The rest of the paper is organized as follows:
Some preliminaries and notations that we use throughout the paper appear in Section \ref{section:preliminaries}.
In Section \ref{section:tree_size} we prove the upper bound of $O(\log{n})$ on 
the size of the query tree in the case of bounded and binomially distributed degrees.
In section \ref{hypergraph}, we use this analysis to give improved algorithms for hypergraph $2$-coloring and $k$-CNF.
In Section \ref{section:matching} we give an LCA for finding a maximal matching in graphs of bounded degree. 
Section \ref{section:load_balancing} expands our query tree result to a special case of bipartite graphs;
we use this bound for bipartite graph to convert online algorithms for balls and bins into LCAs for the same problems.
The appendices provide in-depth discussions of the hypergraph $2$-coloring and analogous $k$-CNF LCAs, and a lower bound to the query tree size.

\section{Preliminaries}\label{section:preliminaries}

Let $G=(V, E)$ be an undirected graph.
We denote by $N_{G}(v)=\{u\in V(G): (u,v)\in E(G)\}$ 
the neighbors of vertex $v$, and by 
$deg_{G}(v)$ we denote the degree of $v$.
When it is clear from the context, we omit the $G$ in the subscript.
Unless stated otherwise, all logarithms in this paper are to the base $2$.
We use $[n]$ to denote the set $\{1,\ldots, n\}$, where $n\geq 1$ is a natural number. 

We present our model of local computation algorithms (LCAs):
Let $F$ be a computational problem and $x$ be an input to $F$.
Let $F(x) = \{y~|~ y {\rm ~is~ a~} {\rm valid~solution~}$ ${\rm for~input~} x\}$.
The {\em search problem} for $F$ is to find any $y \in F(x)$.

A {\em $(t(n), s(n), \delta(n))$-local computation algorithm} $\mcA$
is a (randomized) algorithm which solves a search problem for $F$ for an input $x$ of size $n$.
However, the LCA $\mcA$ does not output a solution $y\in F(x)$, but rather
implements query access to $y \in F(x)$. $\mcA$ receives a sequence of queries $i_1,\ldots,i_q$ and for any $q>0$ satisfies the following:
(1) after each query $i_j$ it produces an output $y_{i_j}$,
(2) With probability at least $1-\delta(n)$ $\mcA$ is \emph{consistent}, that is, the outputs $y_{i_1},\ldots, y_{i_q}$ are substrings of some $y \in F(x)$.
(3) $\mcA$ has access to a random tape and local computation memory on which
it can perform current computations as well as store and retrieve information from previous computations.

We assume that the input $x$, the local computation tape and any
random bits used are all presented in the RAM word model, i.e.,
$\mcA$ is given the ability to access a word of any of these in one step.
The running time of $\mcA$ on any query is at most $t(n)$, which is sublinear in $n$,
and the size of the local computation memory of $\mcA$ is at most $s(n)$.
Unless stated otherwise, we always assume that the error parameter
$\delta(n)$ is at most some constant, say, $1/3$.
We say that  $\mcA$ is a {\em strongly local computation algorithm}
if both $t(n)$ and $s(n)$ are upper bounded by $O(\log^{c}n)$ for some constant $c$.

Two important properties of LCAs are as follows.
We say an LCA $\mcA$ is \emph{query order oblivious} (\emph{query oblivious} for short)
if the outputs of $\mcA$ do not depend on the order of the queries but
depend only on the input and the random bits generated on the random tape of $\mcA$.
%
We say an LCA $\mcA$ is \emph{parallelizable} if $\mcA$
supports parallel queries, that is $\mcA$ is able to answer multiple queries
simultaneously so that all the answers are consistent.

\section{Bounding the size of a random query tree}
\label{section:tree_size}

\subsection{The problem and our main results}\label{Sect:query_tree}

In online algorithms, queries arrive in some unknown order, and the reply to each query depends only on previous queries (but not on any future events). The simplest way to transform online algorithms to LCAs is to process the queries in the order in which they arrive. This, however, means that we have to store the replies to all previous queries, so that even if the time to compute each query is polylogarithmic, the overall space is linear in the number of queries. Furthermore, this means that the resulting LCA is not query-oblivious. The following solution can be applied to this problem (\cite{NO08} and \cite{ARV+11}): Each query $v$ is assigned a random number,  $r(v) \in [0,1]$, called its \emph{rank}, and the queries are performed in ascending order of rank. Then, for each query $x$, a query tree can be constructed, to represent the queries on which $x$ depends. If we can show that the query tree is small, we can conclude that each query does not depend on many other queries, and therefore a small number of queries need to be processed in order to reply to query $x$.  We formalize this as follows:

Let $G=(V,E)$ be an undirected graph. The vertices of the graph represent queries, and the edges represent the dependencies between the queries.
A real number $r(v) \in [0,1]$ is assigned independently and uniformly at random to every vertex
$v \in V$; we call $r(v)$ the \emph{rank} of $v$. This models the random permutation of the vertices.
Each vertex $v \in V$  holds an input $x(v) \in R$, where the range $R$ is some finite set. The input is the content of the query associated with $v$.
A randomized function $F$ is defined inductively on the vertices of $G$ such that
$F(v)$ is a (deterministic) function of $x(v)$ as well as the values of $F$ at the neighbors $w$ of $v$
for which $r(w) < r(v)$. $F$ models the output of the online algorithm.
We would like to upper bound the number of queries to vertices in the graph needed in order to compute $F(v_0)$ for any vertex $v_0 \in G$, namely, the time to simulate the output of query $v_0$ using the online algorithm.

To upper bound the number of queries to the graph, we turn to a simpler task of bounding the size
of a certain $d$-regular tree, which is an upper bound on the number of queries.
Consider an infinite $d$-regular tree $\mathcal{T}$ rooted at $v_{0}$.
Each node $w$ in $\mathcal{T}$ is assigned independently and uniformly at random a real number
$r(w)\in [0,1]$.
For every node $w$ other than $v_{0}$ in $\mathcal{T}$,
let $\text{parent}(w)$ denote the \text{parent} node of $w$.
We grow a (possibly infinite) subtree $T$ of $\mathcal{T}$ rooted at $v$ as follows:
a node $w$ is in the subtree $T$ if and only if $\text{parent}(w)$ is in $T$ and
$r(w)<r(\text{parent}(w))$ (for simplicity we assume all the ranks are distinct real numbers).
That is, we start from the root $v_0$, add all the children of $v_0$ whose ranks are smaller than that
of $v_0$ to $T$. We keep growing $T$ in this manner where a node $w'\in T$ is a leaf node in $T$
if the ranks of its $d$ children are all larger than $r(w')$.
We call the random tree $T$ constructed in this way a \emph{query tree} and
we denote by $|T|$ the random variable that corresponds to the size of $T$.
Note that $|T|$ is an upper bound on the number of queries since
each node in $T$ has at least as many neighbors as that in $G$ and
if a node is connected to some previously queried nodes, this can only decrease
the number of queries. Therefore the number of queries is bounded by the size of $T$.
Our goal is to find an upper bound on $|T|$ which holds with high probability.

We improve the upper bound on the query tree of $O(\log^{d+1}{N})$ given in \cite{ARV+11}  for the case when the degrees are bounded by a constant $d$
and extend our new bound to the case that the degrees of $G$ are binomially distributed, independently and identically with expectation $d$, i.e.,
$deg(v) \sim B(n, d/n)$.

Our main result in this section is bounding, with high probability, the size of the query tree $T$ as follows.


\begin{lemma}
\label{lemma:boundomial}
Let $G$ be a graph whose vertex degrees are bounded by $d$ or distributed independently and identically from the binomial distribution:  $deg(v) \sim B(n,d/n)$. Then there exists a constant 
$C(d)$ which depends only on $d$, such that
\begin{center}
$\Pr[|T|>C(d)\log{n}]<1/n^2$,
\end{center}
where the probability is taken over all the possible permutations $\pi \in \Pi$ of the vertices of $G$, and $T$ is a random query tree in $G$ under $\pi$.
\end{lemma}

\subsection{Overview of the proof}

Our proof of Lemma \ref{lemma:boundomial} consists of two parts. Following~\cite{ARV+11}, we partition the query tree into levels.
The first part of the proof is an upper bound on the size of a single (sub)tree on any level.
For the bounded degree case, this was already proved in ~\cite{ARV+11} (the result is restated as Proposition~\ref{lemma:bounded}).


We extend the proof of ~\cite{ARV+11} to the binomially distributed degrees case. 
In both cases the bound is that with high probability each subtree is of size at most logarithmic in the size of the input.

The second part, which is a new ingredient of our proof,
inductively upper bounds the number of vertices on each level, as the levels increase.
For this to hold, it crucially depends on the fact that all subtrees are generated independently and that the probability of any subtree being large is exponentially small.
The main idea is to show that although each subtree, in isolation, can reach a logarithmic size, their combination is not likely to be much larger. We use the distribution of the sizes of the subtrees, in order to bound the aggregate of multiple subtrees.


\subsection{Bounding the subtree size}


As in \cite{ARV+11}, we partition the query tree into levels and then upper bound
the probability that a subtree is larger than a given threshold.
Let $\mygamma>1$ be a function of $d$ to be determined later.
First, we partition the interval [0,1] into $\mygamma$ sub-intervals:
$I_i = (1-\frac{i}{\mygamma+1}, 1-\frac{i-1}{\mygamma+1}]$,
for $i=1,2,\cdots, \mygamma$ and $I_{\mygamma+1} = [0, \frac{1}{\mygamma+1}]$.
We refer to interval $I_i$ as \emph{level} $i$.
A vertex $v \in T$ is said to be on level $i$ if $r(v) \in I_i$.
We consider the worst case, in which $r(v_0) \in I_1$.
In this case, the vertices on level $1$ form a tree $T_1$ rooted at $v_0$.
Denote the number of (sub)trees on level $i$  by $t_i$.
The vertices on level $2$ will form a set of trees $\{T_2^{(1)}, \cdots, T_2^{(t_2)}\}$,
where the total number of subtrees is at most the sum of the children of all the vertices in $T_1$
(we only have inequality because some of the children of the vertices of $T_1$
may be assigned to levels $3$ and above.)
The vertices on level $i>1$ form a set of subtrees $\{T_i^{(1)}, \cdots T_i^{(t_{i})}\}$.
Note that all these subtrees $\{T_{i}^{(j)}\}$ are generated independently by the same stochastic process, as the ranks of all nodes in $\mathcal{T}$ are i.i.d. random variables. In the following analysis, we will set $\mygamma = 3d$.


For the bounded degree case, bounding the size of the subtree follows from \cite{ARV+11}:



\begin{proposition} [\cite{ARV+11}] \footnote{In \cite{ARV+11}, this lemma is proved for the case of $\mygamma=d+1$. This immediately establishes Proposition \ref{lemma:bounded}, since the worse case is $\mygamma=d+1$.}
\label{lemma:bounded}
Let $L\geq d+1$ be a fixed integer
and let $\mathcal{T}$ be the $d$-regular infinite query tree.
Then for any $1\leq i \leq \mygamma$ and $1 \leq j \leq t_i$,
$\Pr[|T_i^{(j)}|\geq n]\leq \sum_{i=n}^{\infty}2^{-ci}\leq 2^{-\Omega(n)}$,
for all $n \geq \beta$, where $\beta$ is some constant.
In particular, there is an absolute constant $c_{0}$ depending on $d$ only such that
for all $n\geq 1$,
\[
\Pr[|T_i^{(j)}|\geq n]\leq e^{-c_{0}n}.
\]
\end{proposition}


\subsubsection{The binomially distributed degrees case}
\label{subtree_app}

We are interested in bounding the subtree size also in the case that the degrees are not a constant $d$, but rather selected independently and identically from a binomial distribution with mean $d$.
\begin{proposition}
\label{lemma:binomial}
 Let $\mathcal{T}$ be a tree with vertex degree distributed i.i.d. binomially with $deg(v) \sim B(n,d/n)$. For any $1 \leq i \leq \mygamma$ and any $1 \leq j \leq t_i$,
$\Pr[|T_i^{(j)}|\geq n]\leq \sum_{i=n}^{\infty}2^{-ci}\leq 2^{-\Omega(n)}$,
 for $n \geq \beta$, for some constant $\beta > 0$.
\end{proposition}


\begin{proof}
The proof of Proposition~\ref{lemma:binomial} is similar to the proof of Proposition ~\ref{lemma:bounded} in \cite{ARV+11};
we employ the theory of Galton-Watson processes.
For a good introduction to Galton-Watson branching processes see e.g. \cite{Har63}.

Consider a Galton-Watson process defined by the probability function
$\mathbf{p}:=\{p_{k}; k=0,1,2, \ldots \}$, with $p_{k}\geq 0$ and $\sum_{k}p_{k}=1$.
Let $f(s)=\sum_{k=0}^{\infty}p_{k}s^{k}$
be the generating function of $\mathbf{p}$.
For $i=0, 1, \ldots,$ let $Z_{i}$ be the number of offsprings in the $\ith{i}$ generation.
Clearly $Z_{0}=1$ and $\{Z_{i}: i=0, 1, \ldots \}$ form a Markov chain.
Let $m:=\Exp[Z_{1}]=\sum_{k}kp_{k}$ be the expected number of children of any individual.
Let $Z=Z_{0}+Z_{1}+\cdots$ be the sum of all offsprings in all generations of the Galton-Watson process.
The following result of Otter is useful in bounding the probability that $Z$ is large.

\begin{theorem}[\cite{Ott49}]\label{thm:Otter}
Suppose $p_{0}>0$ and that there is a point $a>0$ within the circle of convergence of $f$
for which $af'(a)=f(a)$. Let $\alpha=a/f(a)$.
Let $t=\mathrm{gcd}\{r: p_{r}>0\}$, where $\mathrm{gcd}$ stands for greatest common divisor.
Then

\begin{align*} \label{eqn:Otter}
\Pr[Z=n]=\begin{cases}
  t\left(\frac{a}{2\pi \alpha f''(a)} \right)^{1/2}\alpha^{-n}n^{-3/2} +O(\alpha^{-n}n^{-5/2}),
    &  \text{if $n \equiv 1 \pmod{t}$;} \\
  0, &  \text{if $n \not\equiv 1 \pmod{t}$.}
\end{cases}
\end{align*}
In particular, if the process is \emph{non-arithmetic}, i.e. $\mathrm{gcd}\{r: p_{r}>0\}=1$,
and $\frac{a}{\alpha f''(a)}$ is finite, then
\[
\Pr[Z=n]=O(\alpha^{-n}n^{-3/2}),
\]
and consequently $\Pr[Z\geq n]=O(\alpha^{-n})$.
\end{theorem}


We prove Proposition \ref{lemma:binomial} for the case of tree $T_{1}$ -- the proof actually applies
to all subtrees $T_i^{(j)}$.
Recall that $T_{1}$ is constructed inductively as follows:
for $v \in N(v_{0})$ in $\mathcal{T}$, we add $v$ to $T_{1}$ if
$r(v)<r(v_{0})$ and $r(v) \in I_{1}$.
Then for each $v$ in $T_{1}$, we add the neighbors $w \in N(v)$ in $\mathcal{T}$
to $T_{1}$ if $r(w)<r(v)$ and $r(w) \in I_{1}$.
We repeat this process until there is no vertex that can be added to $T_{1}$.

Once again, we work with the worst case that $r(v_{0})=1$.
To upper bound the size of $T_{1}$, we consider a related random process which also
grows a subtree of $T$ rooted at $v_{0}$, and denote it by $T'_{1}$.
The process that grows $T'_{1}$ is the same as that of $T_{1}$
except for the following difference: if $v\in T'_{1}$ and $w$ is a child vertex of $v$ in $\mathcal{T}$,
then we add $w$ to $T'_{1}$ as long as $r(w) \in I_{1}$. In other words, we give up the requirement that $r(w)<r(v)$.
Clearly, we always have $T_{1}\subseteq T'_{1}$ and hence $|T'_{1}|\geq |T_{1}|$.\\

Note that the random process that generates $T'_{1}$ is in fact a Galton-Watson process,
as the rank of each vertex in $T$ is independently and uniformly distributed in $[0,1]$.
We take vertex $v$ to be the parent node. Since $|I_{1}|=1/\mygamma$,
then for any vertex $u\in V(G)$, $u\neq v$,
the probability that $u$ is a child node of $v$ in $T'_{1}$ is
\[
d/n \cdot 1/\mygamma = d/n\mygamma,
\]
as the random process that connects $w$ to $v$ and
the random process that generates the rank of $w$ are independent
(each edge is chosen with probability $d/n$, and the probability that $r(w)$ is in $I_v = 1/\mygamma$).
It follows that we have a binomial distribution for the number of child nodes of $v$ in $T'_{1}$:
\[
\mathbf{p}=\{(1-q)^n, \binom{n}{1}q(1-q)^{n-1}, \binom{n}{2}q^{2}(1-q)^{n-2}, \ldots, q^{n}\},
\]
where $q:=d/n\mygamma$ is the probability that a child vertex in $T$
appears in $T'_{1}$ when its parent vertex is in $T'_{1}$.
Note that the expected number of children of a vertex in $T'_{1}$ is $nq=d/\mygamma<1$,
so from the classical result on the extinction probability of Galton-Watson processes (see e.g. \cite{Har63}),
the tree $T'_{1}$ is finite with probability one.

The generating function of $\mathbf{p}$ is
\[
f(s)=(1-q+qs)^{n},
\]
as the probability function $\{p_{k}\}$ obeys the binomial distribution $p_{k}=\Pr[X=k]$
where $X\sim B(n,q)$.
In addition, the convergence radius of $f$ is $\rho=\infty$
since $\{p_{k}\}$ has only a finite number of non-zero terms.

\[
f'(s) = nq(1-q+qs)^{n-1}
\]

Solving the equation $af'(a)=f(a)$ yields
$anq(1-q+qa)^{n-1} =  (1-q+qa)^{n}$ and hence
$anq = 1-q+qa $.
Consequently, solving for $a$ gives
\begin{align*}
a& = \frac{1-q}{q(n-1)}\\
& = \frac{1-d/n\mygamma}{d(n-1)/n\mygamma}\\
& = \frac{n\mygamma - d}{d(n-1)} \\
& = \frac{3n-1}{n-1}.
\end{align*}

We can lower bound $\alpha$ as
\begin{align*}
\alpha& = \frac{a}{f(a)}=\frac{1}{f'(a)}\\
& =  \frac{\mygamma}{d (1-q+q\frac{3n-1}{n-1})^{n-1}}\\
& = \frac{3}{\left(1+\frac{2}{3(n-1)}\right)^{n-1}}\\
& \geq  \frac{3}{e^{2/3}} > 1.
\end{align*}

Finally we calculate $f''(a)$:
\begin{align*}
f''(a) &= q^2n(n-1)(1-q + qa)^{n-2}\\
&= \frac{d^2 n(n-1)}{n^2\mygamma^2}
  \left(1-\frac{d}{n\mygamma} + \frac{d}{n\mygamma}\cdot \frac{3n-1}{n-1} \right)^{n-2}\\
&=\frac{n-1}{9n}\left(1+\frac{2}{3(n-1)}\right)^{n-2} \\
&=\Theta(1),
\end{align*}
therefore $\frac{a}{\alpha f''(a)}$ is a bounded constant.



Now applying Theorem~\ref{thm:Otter} to the Galton-Watson process
which generates $T'_{1}$ (note that $t=1$ in our case) gives that, there exists a constant $n_{0}$
such that for $n>n_{0}$,
$\Pr[|T'_{1}|=n]\leq 2^{-cn}$ for some constant $c>0$.
It follows that $\Pr[|T'_{1}|\geq n]\leq \sum_{i=n}^{\infty}2^{-ci}\leq 2^{-\Omega(n)}$
for all $n>n_{0}$.
Hence for all large enough $n$, with probability at least $1-1/n^3$, $|T_{1}|\leq |T'_{1}|=O(\log n)$. \qed
\end{proof}

The following corollary stems directly from Propositions  \ref{lemma:bounded} and \ref{lemma:binomial}:
\begin{corollary}\label{corr:boundomial}
Let $\mathcal{T}$ be  any infinite $d$-regular query tree or tree with vertex degree distributed i.i.d. binomially with $deg(v) \sim B(n,d/n)$.
For any $1 \leq i \leq \mygamma$ and any $1 \leq j \leq t_i$,
with probability at least $1-1/n^3$, $|T_i^{(j)}| = O(\log{n})$.
\end{corollary}

%
%

\subsection{Bounding the increase in subtree size as we go up levels}

\label{section:increase}

From Corollary \ref{corr:boundomial} we know that the size of any subtree, in particular $|T_1|$,
is bounded by $O(\log n)$ with probability at least $1-1/n^3$
in both the degree $d$ and the binomial degree cases.
Our next step in proving
Lemma \ref{lemma:boundomial} 
is to show that, as we increase the levels,
the size of the tree does not increase by more than a constant factor for each level.
That is, there exists an absolute constant $\eta$ depending on $d$ only
such that if the number of vertices on level $k$ is at most $|T_k|$,
then the number of vertices on level $k+1$, $|T_{k+1}|$ satisfies
$|T_{k+1}|\leq \eta \sum_{i=1}^{k}|{T_{i}}|+O(\log n) \leq  2\eta|T_{k}|+O(\log n) $. 
Since there are $\mygamma$ levels in total, this implies that
the number of vertices on all $\mygamma$ levels is at most
$O((2\eta)^{\mygamma} \log n)=O(\log n)$.

The following Proposition establishes our inductive step.
\begin{proposition}
\label{infinite}
For any infinite query tree $\mathcal{T}$ with constant bounded degree $d$ (or degrees i.i.d. $\sim B(n, d/n)$), for any $1\leq i< \mygamma$,
there exist constants $\eta_1>0$ and $\eta_2>0$
s.t. if $\sum_{j=1}^{t_i}|T_i^{(j)}| $ $\leq \eta_1 \log{n}$
then $\Pr[\sum_{j=1}^{t_{i+1}}|T_{i+1}^{(j)}| \geq \eta_1 \eta_2 \log{n}]$ $ < 1/n^2$ for all $n>\beta$,
for some $\beta>0$.
\end{proposition}

\begin{proof}
Denote the number of vertices on level $k$ by $Z_k$ and let $Y_k = \sum_{i=1}^k{Z_i}$.
Assume that each vertex $i$ on level $\leq k$ is the root of a tree of size $z_i$ on level $k+1$.
Notice that $Z_{k+1} = \sum_{i=1}^{Y_k}z_i$.

By Proposition~\ref{lemma:bounded} (or Proposition~\ref{lemma:binomial}),
there are absolute constants $c_{0}$ and $\beta$ depending on $d$ only such that
for any subtree $T_{k}^{(i)}$ on level $k$ and any $n > \beta$,
$\Pr[|T_{k}^{(i)}| = n] \leq e^{-c_{0}n}$.
Therefore, given  $(z_1, \ldots, z_{Y_k})$,
the probability of the forest on level $k+1$ consisting of exactly trees of size
$(z_1, \ldots, z_{Y_k})$ is at most $\prod_{i=1}^{Y_k}e^{-c_{0} (z_i-\beta)} = e^{-c_{0} (Z_{k+1}-\beta Y_k)}$.

Notice that, given $Y_k$ (the number of nodes up to level $k$), there are at most $\binom{Z_{k+1} + Y_k - 1}{Y_k - 1} $ $< \binom{Z_{k+1} + Y_k}{Y_k}$ vectors $(z_1, \ldots, z_{Y_k})$ that can realize $Z_{k+1}$. 

We want to bound the probability that $Z_{k+1} = \eta Y_k$ for some (large enough) constant $\eta>0$.
We can bound this as follows:
\begin{align}
\Pr[|T_{k+1}|=Z_{k+1}]
& < \binom{Z_{k+1} + Y_k}{Y_k}e^{-c_{0} (Z_{k+1}-\beta Y_k)} \nonumber\\
& < {\left(\frac{e \cdot( Z_{k+1} + Y_k)}{Y_k}\right)}^{Y_k} e^{-c_{0} (Z_{k+1}-\beta Y_k)} \nonumber\\
&= {\left(e (1 + \eta)\right)}^{Y_k} e^{-c_{0} (\eta-\beta) Y_k} \nonumber\\
&= e^{Y_k (-c_{0} (\eta-\beta)+ \ln(\eta+1)+1)} \nonumber\\
& \leq e^{-c_{0} \eta Y_k/2}\nonumber,
\end{align}


It follows that there is some absolute constant $c'$ which depends on $d$ only such that
$\Pr[|T_{k+1}|\geq \eta Y_k] \leq e^{-c'\eta Y_k}$.
That is, if $\eta Y_k=\Omega(\log n)$,
the probability that $|T_{k+1}|\geq \eta Y_k$ is at most $1/n^{3}$.
Adding the vertices on all $\mygamma$ levels and applying the union bound, we
conclude that with probability at most $1/n^{2}$, the size of $T$ is at most $O(\log n)$. \qed
\end{proof}

%
%

\section{\texorpdfstring{Hypergraph $2$-coloring and $k$-CNF}{Hypergraph 2-coloring and k-CNF}}
\label{hypergraph}
We use the bound  on the size of the query tree of graphs of bounded degree to improve the analysis of \cite{ARV+11} for hypergraph $2$-coloring. We also modify their algorithm slightly to further improve the algorithm's complexity.
As the algorithm is a more elaborate version of the algorithm of \cite{ARV+11} and the proof is somewhat long, we only state our main theorem for hypergraph $2$-coloring; we defer the proof to  Appendix \ref{app_hypergraph}.
\begin{theorem}
Let $H$ be a $k$-uniform hypergraph s.t. each hyperedge intersects at most $d$ other hyperedges.
Suppose that $k \geq 16 \log{d} + 19$. \\
Then there exists an $(O(\log^4{n}), $ $O(\log^4{n}), 1/n)$-local computation algorithm which, given $H$ and any sequence of
queries to the colors of vertices $(x_1, $ $x_2, \ldots, x_s)$, 
with probability at least $1-1/n^2$,
returns a consistent coloring for all $x_i$'s which 
agrees with a $2$-coloring of $H$. 
Moreover, the algorithm is query oblivious and parallelizable. 
\end{theorem}

Due to the similarity between hypergraph $2$-coloring and $k$-CNF, we also have the following theorem; the proof is in Appendix \ref{kcnf}.
\begin{theorem}
Let $H$ be a $k$-CNF formula with $k\geq 2$. 
Suppose that each clause intersects no more than $d$ other clauses,
and furthermore suppose that $k \geq 16 \log{d} + 19$.\\
Then there exists a $(O(\log^4{n}), O(\log^4{n}), 1/n)$-local computation algorithm which, given a formula $H$ and any sequence of
queries to the truth assignments of variables $(x_1, x_2, \ldots, x_s)$, 
with probability at least $1-1/n^2$,
returns a consistent truth assignment for all $x_i$'s which agrees with some 
satisfying assignment of the $k$-CNF formula $H$. 
Moreover, the algorithm is query oblivious and parallelizable. 
\end{theorem}

%
%
\section{Maximal matching}
\label{section:matching}

We consider the problem of \emph{maximal matching in a bounded-degree graph}. We are given a graph $G=(V,E)$, where the maximal degree is bounded by some constant $d$, and we need to find a maximal matching.A matching is a set of edges with the property that no two edges share a common vertex. The matching is maximal if no other edge can be added to it without violating the matching property.

 Assume the online scenario in which  the edges arrive in some unknown order. The following greedy online algorithm can be used to calculate a maximal matching: When an edge $e$ arrives, we check whether $e$ is already in the matching. If it is not, we check if any of the neighboring edges are in the matching. If none of them is, we add $e$ to the matching. Otherwise, $e$ is not in the matching.

We turn to the local computation variation of this problem.  We would like to query, for  some edge $e \in E$, whether $e$ is part of some maximal matching. (Recall that all replies must be consistent with some maximal matching).

We use the technique of \cite{ARV+11} to produce an almost $O(\log{n})$-wise independent random ordering on the edges, using a seed length of $O(\log^3{n})$.\footnote{Since the query tree is of size $O(\log{n})$ w.h.p., we don't need a complete ordering on the vertices; an almost $O(\log{n})$-wise independent ordering suffices.} When an edge $e$ is queried, we use a BFS (on the edges) to build a DAG rooted at $e$. We then use the greedy online algorithm on the edges of the DAG (examining the edges with respect to the ordering),  and see whether $e$ can be added to the matching. 

As the query tree is an upper-bound on the size of the DAG,  we derive the following theorem from Lemma \ref{lemma:boundomial}.
\begin{theorem}
\label{thm:maximal}
Let $G=(V,E)$ be an undirected graph with $n$ vertices and maximum degree $d$.
Then there is an 
$(O(\log^{3}{n}), $ $O(\log^{3}{n}), $ $1/n)$ - local computation algorithm 
which, on input an edge $e$, 
decides if $e$ is in a maximal matching.
Moreover, the algorithm gives a consistent maximal matching
for every edge in $G$.
\end{theorem}

%
%

\section{The bipartite case and local load balancing}
\label{section:load_balancing}

We consider a general ``power of $d$ choices'' online algorithm for load balancing. In this setting there are $n$ balls that arrive in an online manner, and $m$ bins.
Each ball selects a random subset of $d$ bins, and queries these bins. (Usually the query is simply the current load of the bin.) Given this information, the ball is assigned to one of the $d$ bins (usually to the least loaded bin). We denote by $\mathcal{LB}$ such a generic algorithm (with a decision rule which can depend in an arbitrary way on the $d$ bins that the ball is assigned to). Our main goal is to simulate such a generic algorithm.

The load balancing problem can be represented by a bipartite graph $G = (\{V,U\}, E)$,
where the balls are represented by the vertices $V$ and the bins by the vertices $U$.
The random selection of a bin $u\in U$  by a ball $v\in V$ is represented by an edge.
By definition, each ball $v\in V$ has degree $d$. 
Since there are random choices in the algorithm $\mathcal{LB}$ we need to specify what we mean by a simulation. For this reason we define the input to be the following:
a graph $G =$
 $(\{V, U\},$ $E)$, where $|V| = n $, $|U|=m$, and $n=cm$ for some constant $c\geq 1$. We also allocate a rank $r(u) \in [0,1]$ to every $u \in U$. This rank represents the ball's arrival time: if $r(v)<r(u)$ then vertex $v$ arrived before vertex $u$. Furthermore, all vertices can have an input value $x(w)$. (This value represents some information about the node, e.g., the weight of a ball.)
Given this input, the algorithm $\mathcal{LB}$ is deterministic, since the arrival sequence is determined by the ranks, and the random choices of the balls appear as edges in the graph. Therefore by a simulation we will mean that given the above input, we generate the same allocation as $\mathcal{LB}$.



We consider the following stochastic process: 
Every vertex $v \in V$ uniformly and independently at random chooses $d$ vertices in $U$.
Notice that from the point of view of the bins,
the number of balls which chose them is distributed binomially with $X \sim B(n, d/m)$.
Let $X_v$ and $X_u$ be the random variables for the number of neighbors of vertices $v \in V$ and $u \in U$ respectively.
By definition, $X_v=d$, since all balls have $d$ neighbors, and hence
each $X_u$ is independent of all $X_v$'s.
However, there is a dependence between the $X_u$'s (the number of balls connected to different bins).
Fortunately this is a classical example where the random variables are negatively dependent (see e.g. \cite{DR96}).
\footnote{We remind the reader that two random variables $X_1$ and $X_2$ are negatively dependent if
$\Pr[X_1>x|X_2 = a] < \Pr[X_1>x|X_2 = b]$, for $a>b$ and vice-versa.}

\subsection{The bipartite case}
Recall that in Section \ref{section:tree_size},
we assumed that the degrees of the vertices in the graph were independent.
We would like to prove an $O(\log{n})$ upper bound on the query tree $T$ for our bipartite graph.  As we cannot use the theorems of Section \ref{section:tree_size}
directly, we show that the query tree is smaller than another query tree which meets the conditions of our theorems.

The query tree for the binomial graph is constructed as follows: a root $v_0 \in V$ is selected for the tree. ($v_0$ is the ball whose bin assignment we are interested in determining.)            Label the vertices at depth $j$ in
the tree by $W_j$. Clearly, $W_0= \{v_0\}$. At each depth $d$, we add vertices one at a time to the tree, from left to right, until the depth is "full" and then we 
move to the next depth. Note that at odd depths ($2j+1$) we add bin vertices and at even depths ($2j$) we add ball vertices.

Specifically, at odd depths ($2j+1$) we add, for each $v\in W_{2j}$ its $d$ neighbors $u \in N(v)$ as children, and mark each by $u$.\footnote{A bin can appear several times in the tree. 
It appears as different nodes, but they are all marked so that we know it is the same bin. Recall that we assume that all nodes are unique, as this assumption can only increase the size of the tree.}
At even depths ($2j$) we add for each node marked by $u\in W_{2j-1}$ all its (ball) neighbors $v\in N(u)$ such that $r(v) < r(parent(u))$, if they have not already been added to the tree. Namely, all the balls that are assigned to $u$ by time

A leaf is a node marked by a bin $u_\ell$ for whom all neighboring balls $v\in N(u_\ell)-\{parent(u_\ell)\}$ have a rank larger than its parent, i.e., $r(v)>r(parent(u_\ell))$. Namely, $parent(u_\ell)$ is the first ball to be assigned to bin $u_\ell$.
 This construction defines a stochastic process $F = \{F_t\}$,
where $F_t$ is (a random variable for) the size of $T$ at time $t$. (We start at $t=0$ and $t$ increases by $1$ for every vertex we add to the tree).

We now present our main lemma for bipartite graphs.
\begin{lemma}
\label{lemma:bipartite}
Let $G=(\{V,U\},E)$ be a bipartite graph, $|V| = n$ and $|U| = m$ and $n= cm$ for some constant $c\geq 1$, such that for each
vertex $v \in V$ there are $d$ edges chosen independently and at random between $v$ and $U$. Then there is a constant
$C(d)$ which depends only on $d$ such that 
\begin{center}
$\Pr[|T|<C(d) \log{n}]>1-1/n^2$,
\end{center}
where the probability is taken over all of the possible permutations $\pi \in \Pi$ of the vertices of $G$, and $T$ is a random query tree in $G$ under $\pi$.
\end{lemma}

%
%
\label{bipartite_proof}

To prove Lemma \ref{lemma:bipartite}, we look at another stochastic process $F'$, which constructs a tree $T'$:
we start with  a root $v'_0$. Label the vertices at depth $j$ in
the tree by $W'_j$. Assign every vertex $y$ that is added to the tree a rank $r(y) \in [0,1]$ independently and uniformly at random. Similarly to $T$, $W'_0= \{v'_0\}$. At odd depths ($2j+1$) we add to each $v'\in W'_{2j}$, $d$ children (from left to right).  At even depths ($2j$) we add to each node $u' \in W'_{2j-1}$, $X'_{u'}$ children, where $X'_{u'} \sim B(n, 2d/m)$ and the $X'_{u'}$ of different nodes are i.i.d. Of the nodes added in this level, we remove all those vertices $y'$ for which $r(y')>r(parent(parent(y')))$.
%


Importantly, the neighbor distributions of the vertices in the tree are independent of each other.
If at any point $T'$ has ``more than half the bins'', i.e., the sum of nodes on odd levels is at least $m/2$, we add $n+m$ bin children of rank $0$ to some even-level node in the tree.

Given a tree $T$ we define $squash(T)$ to be the tree $T$ with the
odd levels deleted, and a node $v$ in level $2j$ is connected to node $v'$ in level $2j+2$ if $v=parent(parent(v'))$.

\begin{lemma}
\label{lemma:T'}
There is a constant $C(d)$ which depends only on $d$ such that for all large enough $n$,
$\Pr[|squash(T')|>C(d)\log{n}]<1/n^2.$
\end{lemma}

Because $d\cdot|squash(T')| \geq |T'| \geq |squash(T')|$, we immediately get the following corollary:

\begin{corollary}
\label{corr:T'}
There is a constant $C(d)$ which depends only on $d$ such that for all large enough $n$,
$\Pr[|T'|>C(d)\log{n}]<1/n^2.$
\end{corollary}

We first make the following claim:
\begin{claim}
$squash(T')$  has vertex degree distributed i.i.d. binomially with $deg(v) $ $\sim B(dn,$ $ 2d/m)$.
\end{claim}

\begin{proof}
Each $v \in W'_{2j}$ has $d$ children, each with degree distributed binomially $ \sim B(n, 2d/m)$. 
 For any independent r.v.'s $Y_1, Y_2, \cdots$ where $\forall i>0$, $Y_i \sim B(n,p)$, we know that
 $\displaystyle\sum_{i=1}^q Y_i \sim B(qn,p)$.  The Claim follows. \qed
\end{proof}

We can now turn to the proof of Lemma  \ref{lemma:T'}:

\begin{proof}

As long as $|squash(T')|< m/2$, the proof of the lemma follows the proof of Lemma \ref{lemma:boundomial}
with slight modifications to constants and will therefore be omitted. 

We notice that $|squash(T')|<m/2$ w.p. at least $1-1/n^2$: the proof of  Lemma \ref{lemma:boundomial} is inductive - we show that at level $I_1$,
the size of the subtree is at most $O(\log n)$, and then bound the increase in tree size
as we move to the next level. 
By the level $I_\mygamma$, $|squash(T')|=O(\log n)$ w.p. at least $1-1/n^2$. Therefore 
it follows that $|squash(T')|<m/2$ w.p. at least $1-1/n^2$. \qed
\end{proof}

Before we can complete the proof of Lemma \ref{lemma:bipartite}, we need to define the notion of first order stochastic dominance:

\begin{definition}[First order stochastic dominance]
We say a random variable $X$ \emph{first order stochastically dominates} (\emph{dominates} for short)
a random variable $Y$ if $\Pr[X \geq a] \geq \Pr[Y \geq a]$ for all $a$ and $\Pr[X\geq a] > \Pr[Y\geq a]$ for some $a$.
If $X$ dominates $Y$, then we write $X \geq Y$.
\end{definition}
\begin{lemma}
\label{lemma:stochastic dominance}
For every $t$, $F'_t$ first-order stochastically dominates $F_t$. 
\end{lemma}

\begin{proof}
Assume we add a (bin) vertex $u \in U$ to $T$ at time $t$, the random variable for the number of $u$'s neighbors
is negatively dependent on all other $X_w$, $w \in T_{t} \cap U$. We label this variable $\bar{X} = X_u|\{X_w\}$, $w \in T_{t}$. 

We first show that $F'_t \geq F_t$ when $T$ has less than $m/2$ bins, and then show that $F'_t \geq F_t$ when $T'$ has more than $m/2$ bins. (It is easy to see why this is enough).\\
Assume $|T_{t}\cap U| \leq m/2$. $X_u$ is dependent on at most $m/2$ other random variables, $X_w$. Because the dependency is negative, $X_u$  is maximized when $\forall w, X_w = 0$. Therefore, in the worst case, $X_u$ is dependent on $m/2$ bins with $0$ children. If $m/2$ bins have $0$ children, all edges in $G$ must be distributed between the remaining bins. Therefore $\bar{X} \leq X'_u$, where $X'_u \sim B(n, 2d/m)$.\\
When $T'$ has more than $m/2$ bins, by the construction of $F'_t$, it has more than $m+n$ vertices, and so $F'_t$ trivially dominates $F_t$. \qed
%
\end{proof}

Combining Corollary \ref{corr:T'} and Lemma \ref{lemma:stochastic dominance} completes the proof of Lemma \ref{lemma:bipartite}. 
%
%
\subsection{Local load balancing}
The following theorem states our basic simulation result.

\begin{theorem}\label{thm:local_bb}
Consider a generic online algorithm $\mathcal{LB}$ which requires constant time per query, for $n$ balls and $m$ bins, where $n=cm$ for some constant $c>0$.
There exists an
$(O(\log{n}), O(\log{n}), 1/n)$-local computation algorithm
which, on query of a (ball) vertex $v \in V$,
allocates $v$ a (bin) vertex $u \in U$, such that the resulting allocation is identical to that of $\mathcal{LB}$ with probability at least $1-1/n$.
\end{theorem}

\begin{proof}
Let $K=C(d) \log |U|$ for some constant $C(d)$ depending only on $d$.
$K$ is the upper bound given in Lemma~\ref{lemma:bipartite}.
(In the following we make no attempt to provide the exact values for $C(d)$ or $K$.)

We now describe our $(O(\log n), O(\log n), 1/n)$-local computation algorithm for $\mathcal{LB}$.
A query to the algorithm is a (ball) vertex $v_0 \in V$ and the
algorithm will chose a (bin) vertex   from the $d$ (bin) vertices  connected to $v_{0}$.

We first build a query tree as follows:
Let $v_0$ be the root of the tree. For every $u \in N(u_0)$, add to the tree the neighbors of $u$,
$v \in V$ such that $r(v) < r(v_0)$.
Continue inductively until either $K$ nodes have been added to the random query tree
or no more nodes can be added to it.
If $K$ nodes have been added to the query tree, this is a failure event, and assign to $v_0$ a random bin in $N(v_0)$.
From Lemma~\ref{lemma:bipartite}, this happens with probability at most $1/n^2$,
and so the probability that some failure event will occur is at most $1/n$.
Otherwise, perform $\mathcal{LB}$ on all of the vertices in the tree,
in order of addition to the tree, and output the bin to which ball $v_0$ is assigned to by $\mathcal{LB}$. \qed
\end{proof}

A reduction from various load balancing algorithms gives us the following corollaries to Theorem \ref{thm:local_bb}.

\begin{corollary}  (Using \cite{BCS+06}) 
Suppose we wish to allocate $m$ balls into $n$ bins of uniform capacity, $m \geq n$, where each ball chooses $d$ bins independently and uniformly at random. There exists 
a $(\log{n}, \log{n}, 1/n)$ LCA which allocates the balls in such a way that the load of the most loaded bin is $m/n + O(\log\log{n}/ \log{d})$ w.h.p.
\end{corollary}

\begin{corollary} (Using  \cite{Voc03}) 
Suppose we wish to allocate $n$ balls into $n$ bins of uniform capacity, where each ball chooses $d$ bins independently at random, one from each of $d$ groups of almost equal size $\theta (\frac{n}{d})$. There exists 
a $(\log{n}, \log{n}, 1/n)$ LCA, which allocates the balls in such a way that the load of the most loaded bin is $\ln\ln{n} / (d-1)\ln{2} +O(1)$ w.h.p.
\footnote{In fact, in this setting the tighter bound is
$\frac{\ln\ln{n}} {d\ln{\phi_d}} +O(1)$, where $\phi_d$ is the ratio of the $d$-step Fibonacci sequence, i.e. 
$\phi_d = \lim_{k \rightarrow \infty} \sqrt[k]{F_d(k)}$, 
where for $k < 0 $, $F_d(k) = 0$, $F_d(1) = 1$, and for $k\geq 1$ $F_d(k)=\sum_{i=1}^d{F_d(k-i)}$ }
\end{corollary}

\begin{corollary} (Using  \cite{BBFN10}) 
Suppose we wish to allocate $m$ balls into $n \leq m$ bins, where each bin $i$ has a capacity $c_i$, and $\sum_i c_i = m$. Each ball chooses $d$ bins at random with probability proportional to their capacities. There exists a  $(\log{n}, \log{n}, 1/n)$ LCA which allocates the balls in such a way that the load of the most loaded bin is $2 \log\log{n}+O(1)$ w.h.p.
\end{corollary}

\begin{corollary} (Using  \cite{BBFN10}) 
Suppose we wish to allocate $m$ balls into $n \leq m$ bins, where each bin $i$ has a capacity $c_i$, and $\sum_i c_i = m$. Assume that the size of a large bin is at least $rn \log{n}$, for large enough $r$. Suppose we have $s$ small bins with total capacity $m_s$, and that $m_s = O((n \log{n})^{2/3})$. There exists a  $(\log{n}, \log{n}, 1/n)$ LCA which allocates the balls in such a way that the expected maximum load is less than $5$.
\end{corollary}

\begin{corollary} (Using  \cite{BCM03}) 
Suppose we have $n$ bins, each represented by one point on a circle, and $n$ balls are to be allocated to the bins. Assume each ball needs to choose $d\geq2$ points on the circle, and is associated with the bins closest to these points. There exists a  $(\log{n}, \log{n}, 1/n)$ LCA which allocates the balls in such a way that the load of the most loaded bin is $\ln\ln{n}/ \ln{d} + O(1)$ w.h.p.
\end{corollary}

\subsection{Random ordering}
In the above we assume that we are given a random ranking for each ball.
If we are not given such random rankings
(in fact, a random permutation of the vertices in $U$ will also suffice),
we can generate a random ordering of the balls.
Specifically, since w.h.p. the size of the random query is $O(\log n)$,
 an \emph{$O(\log n)$-wise independent random ordering}\footnote{
 See~\cite{ARV+11} for
the formal definitions of $k$-wise independent random ordering and almost $k$-wise independent random ordering.}
suffices for our local computation purpose.
Using the construction in~\cite{ARV+11} of
$1/n^2$-almost $O(\log n)$-wise independent random ordering over the vertices in $U$
which uses space $O(\log^{3}n)$, we obtain
$(O(\log^3 n), $ $O(\log^3 n),$ $ 1/n)$-local
computation algorithms for balls and bins.

\newpage

\bibliographystyle{plain}
\bibliography{makespan-Bib}

\begin{thebibliography}{10}

\bibitem{Alo91}
N.~Alon.
\newblock A parallel algorithmic version of the {L}ocal {L}emma.
\newblock {\em Random Structures and Algorithms}, 2:367--378, 1991.

\bibitem{ARV+11}
N.~Alon, R.~Rubinfeld, S.~Vardi, and N.~Xie.
\newblock Space-efficient local computation algorithms.
\newblock In {\em Proc.\ 23rd {ACM}-{SIAM} {S}ymposium on {D}iscrete
  {A}lgorithms}, pages 1132--1139, 2012.

\bibitem{ABK+99}
Y.~Azar, A.~Z. Broder, A.~R. Karlin, and E.~Upfal.
\newblock Balanced allocations.
\newblock {\em SIAM Journal on Computing}, 29(1):180--200, 1999.

\bibitem{Bec91}
J.~Beck.
\newblock An algorithmic approach to the {L}ov{\'{a}}sz {L}ocal {L}emma.
\newblock {\em Random Structures and Algorithms}, 2:343--365, 1991.

\bibitem{BBFN10}
P.~Berenbrink, A.~Brinkmann, T.~Friedetzky, and L.~Nagel.
\newblock Balls into non-uniform bins.
\newblock In {\em Proceedings of the 24th IEEE International Parallel and
  Distributed Processing Symposium (IPDPS)}, pages 1--10. IEEE, 2010.

\bibitem{BCS+06}
P.~Berenbrink, A.~Czumaj, A.~Steger, and B.~V{\"o}cking.
\newblock Balanced allocations: The heavily loaded case.
\newblock {\em SIAM J. Comput.}, 35(6):1350--1385, 2006.

\bibitem{BEY98}
A.~Borodin and Ran El-Yaniv.
\newblock {\em Online Computation and Competitive Analysis}.
\newblock Cambridge University Press, 1998.

\bibitem{BCM03}
John~W. Byers, Jeffrey Considine, and Michael Mitzenmacher.
\newblock Simple load balancing for distributed hash tables.
\newblock In {\em Proc. of Intl. Workshop on Peer-to-Peer Systems(IPTPS)},
  pages 80--87, 2003.

\bibitem{DR96}
D.~Dubhashi and D.~Ranjan.
\newblock Balls and bins: A study in negative dependence.
\newblock {\em Random Structures and Algorithms}, 13:99--124, 1996.

\bibitem{Har63}
T.~Harris.
\newblock {\em The Theory of Branching Processes}.
\newblock Springer, 1963.

\bibitem{MR06}
S.~Marko and D.~Ron.
\newblock Distance approximation in bounded-degree and general sparse graphs.
\newblock In {\em APPROX-RANDOM'06}, pages 475--486, 2006.

\bibitem{MRR01}
M.~Mitzenmacher, A.~Richa, and R.~Sitaraman.
\newblock The power of two random choices: A survey of techniques and results.
\newblock In {\em Handbook of Randomized Computing, Vol. I, edited by P.
  Pardalos, S. Rajasekaran, J. Reif, and J. Rolim}, pages 255--312. Norwell,
  MA: Kluwer Academic Publishers, 2001.

\bibitem{NO08}
H.~N. Nguyen and K.~Onak.
\newblock Constant-time approximation algorithms via local improvements.
\newblock In {\em Proc.\ 49th {A}nnual {IEEE} {S}ymposium on {F}oundations of
  {C}omputer {S}cience}, pages 327--336, 2008.

\bibitem{Ott49}
R.~Otter.
\newblock The multiplicative process.
\newblock {\em Annals of mathematical statistics}, 20(2):206--224, 1949.

\bibitem{PR07}
M.~Parnas and D.~Ron.
\newblock Approximating the minimum vertex cover in sublinear time and a
  connection to distributed algorithms.
\newblock {\em Theoretical Computer Science}, 381(1--3), 2007.

\bibitem{RTVX11}
R.~Rubinfeld, G.~Tamir, S.~Vardi, and N.~Xie.
\newblock Fast local computation algorithms.
\newblock In {\em Proc.\ 2nd {S}ymposium on {I}nnovations in {C}omputer
  {S}cience}, pages 223--238, 2011.

\bibitem{RTVX11b}
R.~Rubinfeld, G.~Tamir, S.~Vardi, and N.~Xie.
\newblock Fast local computation algorithms, 2011.

\bibitem{TW07}
K.~Talwar and U.~Wieder.
\newblock Balanced allocations: the weighted case.
\newblock In {\em Proc.\ 39th Annual ACM Symposium on the Theory of Computing},
  pages 256--265, 2007.

\bibitem{Voc03}
Berthold V\"{o}cking.
\newblock How asymmetry helps load balancing.
\newblock {\em J. ACM}, 50:568--589, July 2003.

\bibitem{YYI09}
Y.~Yoshida, Y.~Yamamoto, and H.~Ito.
\newblock An improved constant-time approximation algorithm for maximum
  matchings.
\newblock In {\em Proc.\ 41st Annual ACM Symposium on the Theory of Computing},
  pages 225--234, 2009.

\end{thebibliography}
\newpage

\appendix

%
%
\newcommand{\polylog}[1]{\mathrm{polylog}(#1)}
\newcommand{\ttwo}{\log n}

\section{Hypergraph two-coloring}\label{Sec:hypergraph}
\label{app_hypergraph}
Recall that a \emph{hypergraph} $H$ is a pair $H = (V,E)$ where $V$ is a finite set whose elements are
called \emph{nodes} or \emph{vertices}, and $E$ is a family of non-empty subsets of $V$, 
called \emph{hyperedges}. 
A hypergraph is called \emph{$k$-uniform} if each of its
hyperedges contains precisely $k$ vertices.
A \emph{two-coloring} of a hypergraph $H$ is a mapping $\mathbf{c}: V\to \{\text{red, blue}\}$
such that no hyperedge in $E$ is monochromatic.
If such a coloring exists, then we say $H$ is \emph{two-colorable}.
We assume that each
hyperedge in $H$ intersects at most $d$ other hyperedges.
Let $n$ be the number of hyperedges in $H$. Here we think of $k$ and $d$ as fixed constants
and all asymptotic forms are with respect to $n$.
By the Lov{\'{a}}sz Local Lemma,
when $e(d+1) \leq 2^{k-1}$, the hypergraph $H$ is
two-colorable (e.g. \cite{Alo91}).

Following~\cite{RTVX11b}, we let $m$ be the total number of vertices in $H$. 
Note that $m\leq kn$, so $m=O(n)$.
For any vertex $x\in V$, we use $\mathcal{E}(x)$ to denote the set of hyperedges $x$ belongs to.
For any hypergraph $H = (V,E)$,
we define a \emph{vertex-hyperedge incidence matrix} $\mathcal{M}\in \{0,1\}^{m\times n}$
so that, for every vertex $x$ and every hyperedge $e$, 
$\mathcal{M}_{x,e}=1$ if and only if $e\in \mathcal{E}(x)$.
Because we assume both $k$ and $d$ are constants, 
the incidence matrix $\mathcal{M}$ is necessarily very sparse. 
Therefore, we further assume that the matrix $\mathcal{M}$ is implemented via
linked lists for each row (that is, vertex $x$) and each column (that is, hyperedge $e$). 

Let $G$ be the \emph{dependency graph} of the hyperedges in $H$. 
That is, the vertices of the undirected graph $G$
are the $n$ hyperedges of $H$ and a hyperedge $E_{i}$ is connected to
another hyperedge $E_{j}$ in $G$ if $E_{i}\cap E_{j} \neq \emptyset$.
It is easy to see that if the input hypergraph is given in the 
above described representation, then we can find all the neighbors of any hyperedge $E_{i}$ 
in the dependency graph $G$ (there are at most $d$ of them) in constant time (which depends on $k$ and $d$).

A natural question to ask is:
Given a two-colorable hypergraph $H$, and a vertex $v$,
can we \emph{quickly} compute the coloring of $v$?
Alon et al. gave (\cite{ARV+11}) a $\mathrm{polylog}(n)$-time and space LCA
based on Alon's 3-phase parallel hypergraph coloring algorithm (\cite{Alo91}),
where the exponent of the logarithm depends on $d$.
We get rid of the dependence on $d$ (in the exponent of the logarithm)
using the improved analysis of the query tree in section \ref{section:tree_size},
together with a modified 4-phase coloring algorithm.

Our main result in this section is,
given a two-colorable hypergraph $H$ whose two-coloring scheme is guaranteed by
the Lov{\'{a}}sz Local Lemma (with slightly weaker parameters),
we give a $(O(\log^4{n}), O(\log^4{n}), 1/n)$ - local computation algorithm.
We restate our main theorem:

\begin{theorem}
Let $H$ be a $k$-uniform hypergraph s.t. each hyperedge intersects at most $d$ other hyperedges.
Suppose that $k \geq 16 \log{d} + 19$. \\
Then there exists an $(O(\log^4{n}), $ $O(\log^4{n}), 1/n)$-local computation algorithm which, given $H$ and any sequence of
queries to the colors of vertices $(x_1, $ $x_2, \ldots, x_s)$, 
with probability at least $1-1/n^2$,
returns a consistent coloring for all $x_i$'s which 
agrees with a $2$-coloring of $H$. 
Moreover, the algorithm is query oblivious and parallelizable. 
\end{theorem}

In fact, we only need:
$$ k \geq 3 \lceil \log{16d(d-1)^{3}(d+1)} \rceil + 
			\lceil \log{2e(d+1)} \rceil $$
Throughout the following analysis, we set: 
$ k_1 = k $, and $$ k_i = k_{i-1} - \lceil \log{16d(d-1)^{3}(d+1)} \rceil$$
Notice that the theorem's premise simply implies that $ 2^{k_4-1} \geq e(d+1) $, as required by the Lov{\'{a}}sz Local Lemma.

\subsection{The general phase - random coloring}
In each phase we begin with subsets $V_i$ and $E_i$ of $V$ and $E$,
such that each edge contains at least $k_i$ vertices.
We sequentially assign colors at random to the vertices, as long as every monochromatic edge
has at least $k_{i+1}$ uncolored vertices.
Once the phase is over we do not change this assignment.

If an edge has all of its vertices besides $k_{i+1}$ colored in one color, it is labeled \emph{dangerous}. 
All the uncolored vertices in a dangerous edge are labeled \emph{saved} and we do not color them in this phase.
We proceed until all vertices in $V_i$ are either red, blue, or saved.
Let the \emph{survived} hyperedges be all the edges that do not contain both red and blue vertices.
Each survived edge contains some vertices colored in one color, and at least $k_{i+1}$ saved vertices.

Let $S_i$ be the set of survived edges after a random coloring in Phase $i$, and consider $G|_{S_i}$, the restriction of $G$ to $S_i$
The probability that $G|_{S_i}$ contains a connected component of size $d^3u$ at most $|V_i|2^{-u}$ (\cite{Alo91}).
In particular, after repeating the random coloring procedure $t_i$ times, 
there is no connected component of size greater than $d^3u$ with probability
$$\left(|V_i|2^{-u}\right)^{t_i}$$

If the query vertex $x$ has been assigned a color in the $i$-th phase, we can simply return this color.
Otherwise, if it is a saved vertex we let $C_i(x)$ be the connected component containing $x$ in $G|_{S_i}$.
Finally, since the coloring of $C_i(x)$ is independent of all other uncolored vertices,
we can restrict ourselves to $E_{i+1} = C_i(x)$ in the next phase.


\begin{figure*}
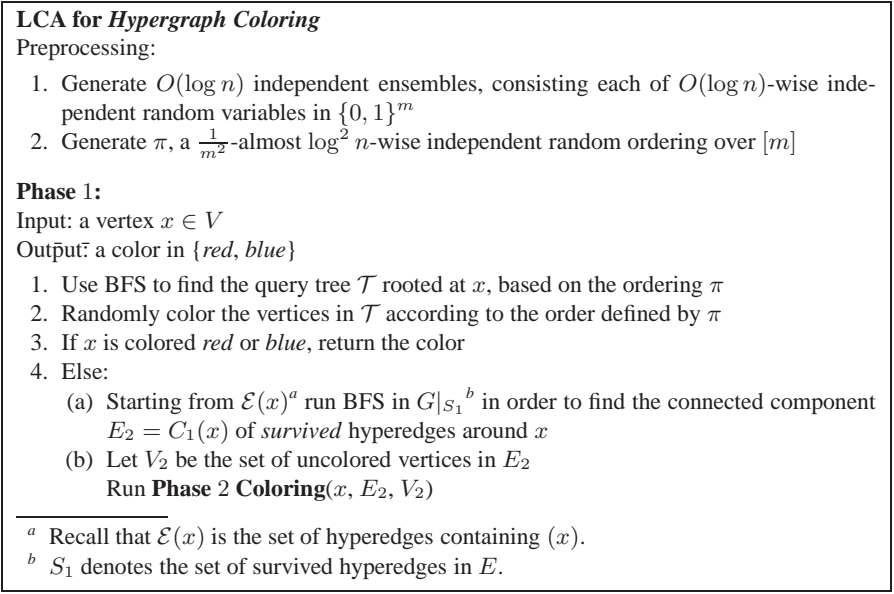

\begin{center}

\fbox{
\begin{minipage}{4.5in}
\small
\textbf{LCA for \emph{Hypergraph Coloring}}\\
Preprocessing:
\vspace{-2mm}
\begin{enumerate}
\item Generate $O(\ttwo)$ independent ensembles, consisting each of 
	$O(\log{n})$-wise independent random variables in $\{0,1\}^{m}$
\item Generate $\pi$, a $\frac{1}{m^{2}}$-almost $\log^2{n}$-wise independent random ordering over $[m]$
\end{enumerate}
\textbf{Phase $1$:}\\
Input: a vertex $x \in V$\\
Out\=put\=: a color in \{\emph{red}, \emph{blue}\}
\vspace{-2mm}
\begin{enumerate}

\item Use BFS to find the query tree $\mathcal{T}$ rooted at $x$, based on the ordering $\pi$
\item Randomly color the vertices in $\mathcal{T}$ according to the order defined by $\pi$
\item If $x$ is colored \emph{red} or \emph{blue}, return the color
\item Else:
\begin{enumerate}
\item Starting from $\mathcal{E}(x)$\footnote{
  Recall that $\mathcal{E}(x)$ is the set of hyperedges containing $(x)$.} 
  run BFS in $G|_{S_1}$\footnote{
  $S_{1}$ denotes the set of survived hyperedges in $E$.}
   in order to find the connected component
   $E_2=C_{1}(x)$ of \emph{survived} hyperedges around $x$
\item Let $V_2$ be the set of uncolored vertices in $E_2$ \\
	 Run \textbf{Phase $2$ Coloring}($x$, $E_2$, $V_2$)
\end{enumerate}
\end{enumerate}
\end{minipage}
}
\end{center}
\caption{Local computation algorithm for \emph{Hypergraph Coloring}}
\label{Fig:coloring}
\end{figure*}


\subsection{Phase 1: partial random coloring}
In the first phase we begin with the whole hypergraph, i.e. $V_1 = V$, $E_1 = E$, and $k_1 = k$. 
Thus, we cannot even assign a random coloring to all the vertices in sublinear complexity.
Instead, similarly to the previous sections, we randomly order the vertices of the hypergraph and use a query tree
to randomly assign colors to all the vertices that arrive before $x$ and may influence it.
Note that this means that we can randomly assign the colors only once.

If $x$ is a saved vertex, we must compute $E_2 = C_1(x)$, 
the connected component containing $x$ in $G|_{S_1}$.
Notice that the size $C_1(x)$ is bounded w.h.p.
$$\Pr\left[ |C_1(x)| > 4d^3\log n \right] < n2^{-4\log n} = n^{-3}$$

In order to compute $C_1(x)$, we run BFS on $G|_{S_1}$.
Whenever we reach a new node, we must first randomly assign colors to the vertices in its query tree,
like we did for $x$'s query tree.
Since (w.h.p.) there are at most $O(\log n)$ edges in $C_1(x)$, we query for trees of vertices in at most $O(\log n)$ edges.
Therefore in total we color at most $O\left(\log^2(n)\right)$ vertices.

Finally, since we are only interested in $O\left(\log^2(n)\right)$ vertices,
we may consider a coloring which is only $O\left(\log^2(n)\right)$-wise independent,
and a random ordering which is only $n^{-3}$-almost $O\left(\log^2(n)\right)$-independent.
Given the construction in ~\cite{ARV+11}, this can be done in space and time complexity $O\left(\log^4(n)\right)$

\begin{figure*}
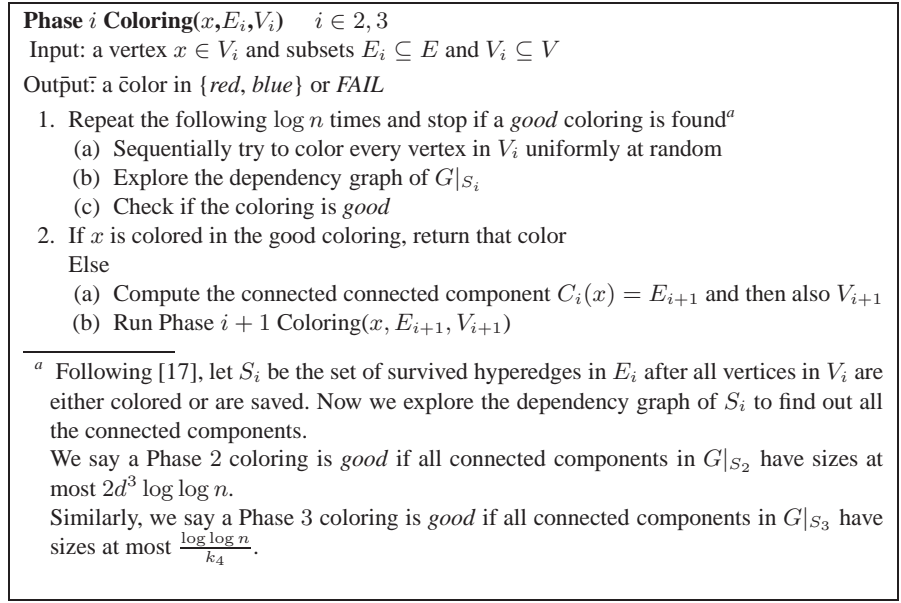

\begin{center}
\fbox{
\begin{minipage}{4.5in}
\small
\textbf{Phase $i$ Coloring($x$,$E_i$,$V_i$)} \hspace{0.1in} $i \in {2,3}$ \\
\vspace{0.7mm}
Input: a vertex $x \in V_i$ and subsets $E_i \subseteq E$ and $V_i \subseteq V$\\
Out\=put\=: a \=color in \{\emph{red}, \emph{blue}\} or \emph{FAIL}
\vspace{-2mm}
\begin{enumerate}
\item Repeat the following $\ttwo$ times and stop if a \emph{good} coloring is found\footnote{
	Following~\cite{RTVX11b}, let $S_{i}$ be the set of survived hyperedges in $E_i$
	after all vertices in $V_i$ are either colored or are saved.
	Now we explore the dependency graph of $S_{i}$ to find out all
	the connected components.\\
	We say a Phase $2$ coloring is \emph{good} if all connected components 
	in $G|_{S_2}$ have sizes at most $2d^3\log\log n$.\\  
	Similarly, we say a Phase $3$ coloring is \emph{good} if all connected components 
	in $G|_{S_3}$ have sizes at most $\frac{\log\log n}{k_4}$.\\ }
\begin{enumerate}
  	\item Sequentially try to color every vertex in $V_i$ uniformly at random
  	\item Explore the dependency graph of $G|_{S_{i}}$
  	\item Check if the coloring is \emph{good}
\end{enumerate}
\item If $x$ is colored in the good coloring, return that color\\
	Else 
	\begin{enumerate}
	\item Compute the connected connected component $C_i(x) = E_{i+1}$ and then also $V_{i+1}$
	\item Run Phase $i+1$ Coloring($x, E_{i+1}, V_{i+1}$)
	\end{enumerate}  
\end{enumerate}  

\end{minipage}
}
\end{center}
\caption{Local computation algorithm for \emph{Hypergraph Coloring}:  Phase $2$ and Phase $3$}
\label{Fig:Phase2}
\end{figure*}


\subsection{Phase 2 and 3: gradually decreasing the component size}
Phase 2 and 3 are simply iterations of the general phase with parameters as described below.
With high probability we have that $|E_2| \leq 4d^3\log n$,
and each edge has $k_2$ uncolored vertices.
After at most $t_2 = \ttwo$ repetitions of the random coloring procedure, 
we reach an assignment that leaves a size $2d^3\log\log n$-connected component of survived edges with probability
$$\left( \left(4d^3k_2\log n\right) 2^{-2\log\log n} \right)
					^{\ttwo} < n^{-3}$$
Similarly, in the third phase we begin with $|E_3| < 2d^3\log\log n$, and after $t_3 = \log n$ repetitions
we reach an assignment that leaves a size $\frac{\log\log n}{k_4}$-connected component of survived edges with probability
$$\left( \left(2d^3k_3\log\log n\right) 2^{-\frac{\log\log n}{d^3k_4}}\right)
					^{\log n} < n^{-3}$$

\begin{figure*}
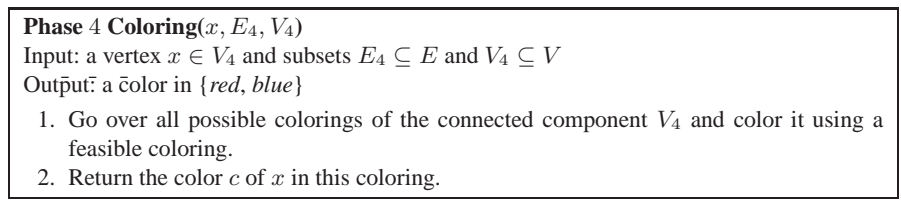

\begin{center}
\fbox{
\begin{minipage}{4.5in}
\small
\textbf{Phase $4$ Coloring($x, E_4, V_4$)} \\
Input: a vertex $x \in V_4$ and subsets $E_4 \subseteq E$ and $V_4 \subseteq V$\\
Out\=put\=: a \=color in \{\emph{red}, \emph{blue}\}
\vspace{-2mm}
\begin{enumerate}

\item Go over all possible colorings of the connected component $V_4$
	and color it using a feasible coloring.
\item Return the color $c$ of $x$ in this coloring.

\end{enumerate}
\end{minipage}
}
\end{center}
\caption{Local computation algorithm for \emph{Hypergraph Coloring}: Phase $4$}
\label{Fig:Phase4}
\end{figure*}


\subsection{Phase 4: brute force}
Finally, we are left with a connected component of $|E_4| < \frac{\log\log n}{k_4}$, 
and each edge has $k_4$ uncolored vertices. 
By the Lovasz Local Lemma, there must exists a coloring (see e.g. Theorem 5.2.1 in ~\cite{Alo91}).
We can easily find this coloring via brute force search in time $O(\log n$).


\section{\texorpdfstring{$k$-CNF}
{k-CNF}}\label{kcnf}

As another application, our hypergraph coloring algorithm can be easily modified to 
compute a satisfying assignment of a $k$-CNF formula, 
provided that the latter satisfies some specific properties.

Let $H$ be a $k$-CNF formula on $m$ Boolean variables $x_{1}, \ldots, x_{m}$.
Suppose $H$ has $n$ clauses $H=A_{1} \wedge \cdots \wedge A_{n}$
and each clause consists of exactly $k$ distinct literals.\footnote{
Our algorithm works for the case that each clause has at least $k$ literals; 
for simplicity, we assume that all clauses have uniform size.}
We say two clauses $A_{i}$ and $A_{j}$ \emph{intersect} 
with each other if they
share some variable (or the negation of that variable).
As in the case for hypergraph coloring, $k$ and $d$ are fixed constants
and all asymptotics are with respect to the number of clauses $n$ (and hence $m$, since $m\leq kn$).
Our main result is the following.

\begin{theorem}
Let $H$ be a $k$-CNF formula with $k\geq 2$. 
Suppose that each clause intersects no more than $d$ other clauses,
and furthermore suppose that $k \geq 16 \log{d} + 19$.\\
Then there exists a $(O(\log^4{n}), O(\log^4{n}), 1/n)$-local computation algorithm which, given a formula $H$ and any sequence of
queries to the truth assignments of variables $(x_1, x_2, \ldots, x_s)$, 
with probability at least $1-1/n^2$,
returns a consistent truth assignment for all $x_i$'s which agrees with some 
satisfying assignment of the $k$-CNF formula $H$. 
Moreover, the algorithm is query oblivious and parallelizable. 
\end{theorem}

\begin{sketch}
We follow a 4-phase algorithm similar to that of 
hypergraph two-coloring as presented in appendix \ref{Sec:hypergraph}.
In every phase, we sequentially assign random values to a subset of the remaining variables,
maintaining a threshold of $k_i$ unassigned variables in each unsatisfied clause.
Since the same (in fact, slightly stronger) bounds that hold for the connected components
in the hyperedges dependency graph also hold for the clauses dependency graph (\cite{RTVX11b}),
we can return an answer which is consistent with a satisfying assignment with probability at least $1-1/n^2$.
\end{sketch}

%
%

\section{Lower bound on the size of the query tree}
\label{section:lower_bound}
We prove a lower bound on the size of the query tree.
\begin{theorem}
\label{thm:lower_bound}
Let $G$ be a random graph whose vertex degree is bounded by $d \geq 2$ or distributed independently and identically from the binomial distribution:  $deg(v) \sim B(n,d/n)$ ($d \geq 2$). Then
\begin{center}
$\Pr[|T|>\log{n} / \log\log{n}]>1/n$,
\end{center}
where the probability is taken over all random permutations $\pi \in \Pi$ of the vertices, and $T$ is the largest query tree in $G$ (under $\pi$).
\end{theorem}\begin{proof}

For both the bounded degree and the binomial distribution cases, there exists a path of length at least $k = \log{n} / \log\log{n}$ in the graph  w.h.p. Label the vertices on the path $v_1, v_2, \ldots, v_k$. There are $k!$ possible permutations of the weights of the vertices on the path. The probability of choosing the permutation in which $w(v_1)<w(v_2)<\ldots<(v_k)$ is $1/k!$.
\begin{align*}
k! & = (\log{n} / \log\log{n})!\\
& < (\log{n} / \log\log{n})^{\log{n} / \log\log{n}}\\
&<n.
\end{align*}
Therefore, $1/k!>1/n$ and so the probability of the query tree having size\\ $\log{n} / \log\log{n}$ is at least $1/n$. \qed

\end{proof}

\end{document}